\documentclass[11pt]{article}

\usepackage{lipsum}

\usepackage[utf8]{inputenc} % set input encoding (not needed with XeLaTeX)

\usepackage{fancyhdr} 
\usepackage{amsthm}

\usepackage[utf8]{inputenc} 
\usepackage[T1]{fontenc}
\usepackage{url}
\usepackage{ifthen}

\usepackage{algpseudocode,algorithm,algorithmicx}

\usepackage{graphicx}
\usepackage{amssymb}
\usepackage{xcolor}
\usepackage{enumitem}
\usepackage{mathtools}
\usepackage{latexsym}
\usepackage{amsmath}
\usepackage{amsfonts}
\usepackage{amsthm}
\usepackage{amssymb}
\usepackage{bbm}
\usepackage{centernot}
\usepackage{tcolorbox}

\usepackage{listings}
\usepackage{color} %red, green, blue, yellow, cyan, magenta, black, white
\definecolor{mygreen}{RGB}{28,172,0} % color values Red, Green, Blue
\definecolor{mylilas}{RGB}{170,55,241}
\usepackage{thm-restate}

\usepackage[colorlinks,citecolor=blue,linkcolor=magenta,bookmarks=true,hypertexnames=false]{hyperref}
\usepackage[nameinlink,capitalise]{cleveref}

\crefname{equation}{Equation}{Equations}
\crefname{lemma}{Lemma}{Lemmata}
\crefname{claim}{Claim}{Claims}
\crefname{theorem}{Theorem}{Theorems}
\crefname{proposition}{Proposition}{Propositions}
\crefname{corollary}{Corollary}{Corollaries}
\crefname{claim}{Claim}{Claims}
\crefname{remark}{Remark}{Remarks}
\crefname{definition}{Definition}{Definitions}
\crefname{fact}{Fact}{Facts}
\crefname{question}{Question}{Questions}
\crefname{condition}{Condition}{Conditions}
\crefname{algorithm}{Algorithm}{Algorithms}
\crefname{assumption}{Assumption}{Assumptions}

\newtheorem{theorem}{Theorem}[section]
\newtheorem{lemma}[theorem]{Lemma}

\newtheorem{corollary}[theorem]{Corollary}
\newtheorem{claim}[theorem]{Claim}

\newtheorem{definition}[theorem]{Definition}
\newtheorem{observation}{Observation}
\newtheorem{fact}[theorem]{Fact}

\theoremstyle{definition}

\newtheorem{remark}[theorem]{Remark}

\newcommand\inner[2]{\langle #1, #2 \rangle}
\newcommand\system[3]{[ \, #1 \mid #2 \, ]_{#3}}

\def\NZT{\otimes_{i=1}^{N} \mathbb{R}_{\neq 0}^{d}}
\def\T{\otimes_{i=1}^{N} \mathbb{R}^{d}}
\def\A{\mathbf{A}}
\def\GF{\mathbb{F}_2}
\def\e{\mathbf{e}}
\def\D{\mathbf{D}}
\def\y{\mathbf{y}}
\def\x{\mathbf{x}}
\def\b{\mathbf{b}}
\def\P{\mathbf{P}}
\def\R{\mathbb{R}}
\def\u{\mathbf{u}}
\def\c{\mathbf{c}}

\def\B{\mathbf{B}}
\def\o{\mathbf{0}_d}

\newcommand\numberthis{\addtocounter{equation}{1}\tag{\theequation}}

\DeclareMathOperator{\rank}{rank}

\DeclareMathOperator{\Row}{row}
\DeclareMathOperator{\vect}{vec}
\DeclareMathOperator{\Span}{span}
\DeclareMathOperator{\sgn}{sign}

\newcommand{\cU}{\mathcal{U}}

\def\colorful{1}

\ifnum\colorful=1

\else

\fi

\allowdisplaybreaks

\title{Simple and Nearly-Optimal Sampling for Rank-1 Tensor Completion via Gauss-Jordan}

\author{
Alejandro Gomez-Leos\thanks{University of Texas at Austin. alexgomezleos@utexas.edu}\\
\and
Oscar L\'opez \thanks{Harbor Branch Oceanographic Institute, Florida Atlantic University. lopezo@fau.edu}\\
}

\begin{document}

\maketitle

\setcounter{page}{0}

\thispagestyle{empty}

\begin{abstract}
We revisit the sample and computational complexity of completing a rank-1 tensor in $\otimes_{i=1}^{N} \mathbb{R}^{d}$, given a uniformly sampled subset of its entries.
We present a characterization of the problem (i.e. nonzero entries) which admits an algorithm amounting to Gauss-Jordan on a pair of random linear systems.
For example, when $N = \Theta(1)$, we prove it uses no more than $m = O(d^2 \log d)$ samples and runs in $O(md^2)$ time.
Moreover, we show any algorithm requires $\Omega(d\log d)$ samples.

By contrast, existing upper bounds on the sample complexity are at least as large as $d^{1.5} \mu^{\Omega(1)} \log^{\Omega(1)} d$, where $\mu$ can be $\Theta(d)$ in the worst case.
Prior work obtained these looser guarantees in higher rank versions of our problem, and tend to involve more complicated algorithms.
    
\end{abstract}

\newpage

\section{Introduction}
Tensor completion is a higher-order generalization of the well-known matrix completion problem~\cite{candes2010power,candes2012exact}.
More precisely, a tensor $\cU \in \T$ is a multi-dimensional array whose entries $\cU_{(i_1, \dots, i_N)}$ are specified by an ordered tuple of $N$ indices, each in $[d]:= \{1, \dots, d\}$.
Stated loosely, the problem is to recover the entirety of a tensor $\cU$ observing only a small subset of its entries.
\\

If $\cU$ is arbitrary, then this task is impossible without observing all $d^N$ entries of $\cU$.
As in matrix completion, the situation becomes interesting when $\cU$ has a \emph{low-rank} structure.
Conceptually, this means $\cU$'s entries are described by an interaction of a small number of variables---far fewer than $d^N$.
In practice, the framework of tensor completion enjoys several applications.
We refer to~\cite{kolda2009tensor, song2019tensor} for a broader overview of these.
\\
 
In this work, we revisit the problem of rank-$1$ tensor completion.
A tensor $\cU \in \T$ is said to be rank-$1$ if there exists $\{\mathbf{u}_1, \dots, \mathbf{u}_N \} \subseteq \mathbb{R}^{d}$ such that
\begin{equation}\label{eqn:polynomial}
    \cU_{(i_1, i_2, \dots, i_N)} = (\mathbf{u}_1)_{i_1} \cdot (\mathbf{u}_2)_{i_2} \cdot \ldots \cdot (\mathbf{u}_N)_{i_N} \quad  \quad \forall \, (i_1, i_2 \dots, i_N) \in [d]^{N},
\end{equation}
or for shorthand, $\cU = \mathbf{u}_1 \otimes \dots \otimes \mathbf{u}_N$.
Perhaps the most fundamental variant, we consider the setting in which the estimation algorithm only has access to uniformly drawn entries and is required to be correct with bounded error probability.\footnote{The optimal dependence for a given failure probability tolerance $\delta \in (0,1)$, even for general rank problems, seems to be unexplored. 
However, one can amplify to a $1-\delta$ guarantee with $O(\log\frac{1}{\delta})$ overhead by the usual method.}
Specifically, we study the following.

\begin{tcolorbox}[colframe=black,colback=white,sharp corners]
\textbf{Problem:} Rank-$1$ Completion

\textbf{Input:} Uniform sampling access to entries of rank-$1$ tensor $\cU$

\textbf{Output:}
Oracle access to $\hat{\cU}$ where $\hat{\cU} = \cU$ with probability $\geq 2/3$
\end{tcolorbox}

We are interested in the sample and computational complexity of the above as dependent on $d$ and $N$.
Our main contribution is a straightforward algorithm which is easiest to describe when the input tensor class is restricted to rank-1 tensors in $\NZT$.
As a consequence, we obtain the following.
\begin{theorem}(informal version of Theorem~\ref{theorem:main-result})\label{theorem:main-informal}
    Assume $\cU$ is an arbitrary rank-1 tensor with nonzero entries, and $N=\Theta(1)$ is a constant independent of $d$.\footnote{$d \gg N$ in practice, with tensors of $N > 5$ rarely seen in applications. For this reason we focus on optimality in $d$.}
    There exists an algorithm that solves \textbf{Rank-1 Completion} which draws at most $m = O(d^2 \log d)$ samples and can be implemented in time $O(md^2)$.
    Moreover, at least $\Omega(d\log d)$ samples are information-theoretically necessary.
\end{theorem}
For analogous problems, it is well-known that these complexities are usually impacted by the \emph{incoherence} $\mu$ of the input tensor class---perhaps surprisingly, our bounds have no such dependence.
This observation seems to have gone mostly undescribed in the current literature.
In effect, this work serves to elucidate the sample complexity gap between the rank-1 and general rank problem variants.
We elaborate on these points after reviewing the notion of $\mu$-incoherent tensors.

\begin{definition}\label{defn:incoherent-subspace}
    An $r$-dimensional subspace $W \subseteq \R^d$ is called $\mu$-incoherent if satisfies 
    \begin{equation*}
        \|\text{Proj}_{W} \, \e_i \|_{2} \leq \sqrt{\frac{\mu r}{d}}  \quad \forall i \in [d].
    \end{equation*}
\end{definition}

\begin{definition}[e.g.~\cite{cai2019nonconvex, singh2020rank}]\label{defn:CP}
    The CP-rank of a tensor $\cU \in \T$ is the minimal integer such that $\mathcal{U} = \sum_{k=1}^{r} \sigma_k \, \mathbf{u}_{1,k} \otimes \dots \otimes \mathbf{u}_{N,k} $ holds for vectors $\{\mathbf{u}_{1,k}, \dots, \mathbf{u}_{N,k}\}_{k=1}^{r} \subseteq \R^d$ and scalars $\{\sigma_k\}_{k=1}^{r} \subseteq \R$.
\end{definition}
\begin{definition}[e.g.~\cite{liu2020tensor,singh2020rank}]\label{defn:incoherent-tensor}
    A (CP) rank-$r$ tensor is called $\mu$-incoherent if each of the following subspaces are $\mu$-incoherent
    \begin{equation}\label{eqn:subspaces}
        \Span \left\{ \frac{\mathbf{u}_{1,k}}{\|\mathbf{u}_{1,k}\|_2} , \frac{\mathbf{u}_{2,k}}{\|\mathbf{u}_{2,k}\|_2} , \dots , \frac{\mathbf{u}_{N,k}}{\|\mathbf{u}_{N,k}\|_2} \right \}, \quad k=1, 2, \dots, N.
    \end{equation}
\end{definition}

\smallskip
In the rank-1 case, $\cU = \mathbf{u}_1 \otimes \dots \otimes \mathbf{u}_N$ being $\mu$-incoherent simply means $\| \mathbf{u}_i\|_{\infty} \leq \sqrt{\mu/d} \quad \forall i \in [d]$.

\paragraph{Prior Work}
Our problem is a very basic case of the well-studied exact tensor completion problem under uniform sampling.
As such, we'll describe the bounds obtained for our case.\footnote{Some of these bounds are for the \emph{expected} sample complexity under the Bernoulli model~\cite{candes2010power}. Rest assured, the comparison is not that imprecise due to concentrations.}
Popularly, the earliest guarantees for matrices ($N=2$) were given in~\cite{candes2010power, recht2011simpler, candes2012exact}.
Notably, these exhibited a strong dependence of $\mu$ on the sample complexity, culminating in~\cite{chen2015incoherence} establishing that $\lesssim  d \mu \log^2 d$ samples suffice for $d^{-\Theta(1)}$ failure probability tolerance.
For cubic tensors ($N=3$), the works~\cite{jain2014provable, xia2019polynomial, liu2020tensor} indicate a scaling increase to $\lesssim  d^{1.5} \mu^{O(1)} \log^{O(1)} d$ for a comparative failure tolerance.
The theory for $N \geq 4$ has received much less attention, with the first results from~\cite{krishnamurthy2013low}, although for adaptively chosen entries. 
These scaled as $\lesssim d\mu^{O(N)} N^{O(1)} \log \frac{d}{\delta}$ for tolerance $d\delta$.
Under varying structural assumptions, the works~\cite{montanari2018spectral,stephan2024non,haselby2024tensor} make significant advancements, but still with dependence on $\mu^{O(1)}$.
Since $\mu$ can be $\Theta(d)$ in the worst case, this substantiates our assertion that Theorem~\ref{theorem:main-informal} elucidates a complexity gap between \textbf{Rank-1 Completion} and these more generic problems.\footnote{Two well-known lower bounds are given in~\cite{candes2010power,singh2020rank} (Theorem 1.7 and Theorem 5, resp.).
These seem to suggest a multiplicative $\mu^{\Omega(N)}$ is missing from our upper bound.
However, their hard instances do not apply to our setting.
}
\\

On the other hand, a distinct series of works~\cite{kiraly2015algebraic, kahle2017geometry,rendon2018algebraic,singh2020rank,zhou2024rank} have focused on the \emph{completability} of a fixed, partially observed rank-1 tensor (i.e. the solution set for the polynomial system~\eqref{eqn:polynomial}).
Perhaps as a testament to the non-triviality of this problem, these works invoke advanced tools from algebraic geometry and matroid theory~\cite{kiraly2015algebraic, rendon2018algebraic}.
We emphasize these do not study the sampling aspect of our problem.
\\

The algorithms of these prior works tend to be quite intricate, albeit specialized for the harder problem of general rank tensor completion.
A few popular methods are alternating minimization~\cite{jain2014provable,xia2019polynomial,liu2020tensor} and semidefinite programming / sum-of-squares hierarchies~\cite{chen2015completing,potechin2017exact,zhou2024rank}, amongst others.
For a more complete overview, we refer to~\cite{cai2019nonconvex,liu2020tensor,haselby2024tensor}.
\\

By contrast, in this work we describe a linear algebraic characterization, which was independently observed by~\cite{singh2020rank} for cubic tensors.
However, the work did not leverage this to study the statistical hardness of our problem.
We show \textbf{Rank-1 Completion} admits an extremely simple algorithm, whose main bottleneck is solving a $\Tilde{O}(d^2) \times \Tilde{O}(d)$ system over $\GF$.
As a consequence, the toolbox for our analysis only consists of elementary linear algebraic and probabilistic arguments.

\paragraph{Our Result}

In this work, we prove the following.
\begin{theorem} \label{theorem:main-result}
    If $\cU \in \NZT$ is a rank-1 tensor with $\| \cU\|_{\infty} \leq \rho$ for some $\rho > 0$, then there exists an algorithm solving \textbf{Rank-1 Completion} using $m \lesssim (dN)^2 \log d$ uniformly drawn entries, and can be implemented in time $\lesssim qN + m(dN)^2$, where $q$ is the number of queried entries of $\mathcal{U}$.
    \smallskip
    
    Moreover, there exists absolute constants $n_0 \in \mathbb{N}$ and $C > 0$ such that $dN \geq n_0$ implies any estimator drawing less than $C  d \log dN$ samples suffers error $\| \hat{\cU} - \cU \|_{F} \gtrsim \rho \sqrt{d^{N-1}}$ with probability $\geq 1/3$.

\end{theorem}

\begin{remark}
    Many lower bounds in the field are stated in a somewhat different sense---that if too few samples are taken, then there likely exists distinct tensors agreeing on the observed entries.
    The standard conclusion is that any algorithm that solely bases its decision on the observed entries must fail~\cite{candes2010power, krishnamurthy2013low}.
    However, this may be unsatisfying since it does not indicate how much error might be suffered, or how well the adversary fares against an estimator with side information.
    This motivates the style of our bound.
\end{remark}

To prove the theorem, we show every rank-1 tensor in $\otimes_{i=1}^{N} \mathbb{R}_{\neq 0 }^d$ is in correspondence with a pair of linear systems---one over $\GF$, and the other over $\mathbb{R}$.
The former system "encodes information" about the signs of the tensor's entries, whereas the latter system does so for their magnitude.
In a precise sense, we show that the "information" revealed by sampling an entry is "equivalent" to the "information" encoded by one row of the associated augmented matrix.
These systems are overdetermined, so it suffices to obtain just a small subset of their rows. 
In essence, this characterization illuminates that the sample complexity can be studied as a matrix row sketch problem over $\GF$.
\\

In what follows, we describe some specialized notation to be used for the rest of this work, and delve into a simple characterization for rank-1 tensors in $\NZT$.

\paragraph{Notation}
We reserve calligraphic letters (e.g. $\mathcal{U}$) for tensors, upper-case bold for matrices (e.g. $\mathbf{U}$), and lower-case bold for vectors (e.g. $\mathbf{u}$). 
% and upper-case plain for random variables (e.g. $X$).
All vectors are conventionally column vectors, unless otherwise stated.
We denote $\e_k$ for $k \in [d]$ as the indicator row vector for $[d]$, interpreting the $\e_{k}$'s having entries in $\GF$ or $\mathbb{R}$, wherever clear from context.
For two row vectors $\x$ and $\y$, we denote $[\x, \y]$ as their concatenation.
For a $n_1 \times n_2$ matrix $\A$ and the vector $\mathbf{b}$ with $n_1$ elements, we denote their induced linear system over a field $\mathbb{K}$ by the augmented matrix $ \system{\A}{\mathbf{b}}{\mathbb{K}}$, regardless of whether the system is consistent or not.
For a scalar-valued function $f$ applied to a tensor, matrix, or vector, we mean entry-wise.
For a matrix $\B$, we denote $\B_i$ as its $i^{th}$ row.
We denote the rowspace of a matrix $\B$ over field $\mathbb{K}$ by $\Row_{\mathbb{K}}(\B)$.
The Frobenius norm of a tensor is given by $\|\cU\|_{F}^2 = \sum_{(i_1, \dots, i_N)} \cU_{(i_1, \dots, i_N)}^2$.

\section{Nonzero Rank-1 Tensors in Exponential Form} 

In this section, we describe a characterization of rank-1 tensors in $\NZT$, which is at the core of Algorithm~\ref{pdcode:exactrank1tc}. 
Suppose that $\cU \in \NZT$ satisfies~\eqref{eqn:polynomial} for some column vectors $\u_1, \dots, \u_N \subseteq \R^d$.
Let us define $\x \in \R^{dN}$ as
\begin{equation} \label{eqn:decision-variable}
    \x := \begin{pmatrix}
         \u_1 \\ \u_2 \\ \vdots \\ \u_N 
    \end{pmatrix} = \begin{pmatrix}
         \begin{pmatrix}
             (\u_1)_1 \\ \vdots \\ (\u_1)_d
         \end{pmatrix} \\ \vdots \\ \begin{pmatrix}
             (\u_N)_1 \\ \vdots \\ (\u_N)_d
         \end{pmatrix}
    \end{pmatrix}.
\end{equation}

Setting this aside, we focus on $\cU$, which we can re-write as follows for all $(i_1,\dots,i_N) \in [d]^N$.
\begin{align*} 
    \cU_{(i_1, i_2, \dots, i_N)} &= \sgn \left( \prod\nolimits_{\ell = 1}^{N} (\mathbf{u}_{\ell})_{i_{\ell}} \right) \, \left| \prod\nolimits_{\ell = 1}^{N} (\mathbf{u}_{\ell})_{i_{\ell}} \right| \\
    &=   \left( \prod\nolimits_{\ell = 1}^{N} \sgn \left( (\mathbf{u}_{\ell})_{i_{\ell}} \right) \right)  \, \left(  \exp \left( \sum\nolimits_{\ell=1}^{N}  \log \left|(\mathbf{u}_{\ell})_{i_{\ell}} \right| \right) \right) \\
    :&= \cU'_{(i_1, i_2, \dots, i_N)} \,  \exp\left(\cU''_{(i_1, i_2, \dots, i_N)}\right) \numberthis \label{eqn:decomp}
\end{align*}
Noting that $\cU'_{(i_1, i_2, \dots, i_N)} = -1$ iff an odd number of the variables $(\mathbf{u}_{1})_{i_1} \dots (\mathbf{u}_{N})_{i_N}$ are negative, we observe that the entry $\cU'_{(i_1, i_2, \dots, i_N)}$ resembles the parity function on the sign of these variables.
Hence, using an appropriate 1-1 transformation $\varphi$\footnote{e.g. $\varphi(z) := -\frac{1}{2}(z-1)$} between $\{\pm 1\}$ and \{0, 1\}, we obtain that the $\{\mathbf{u}_{\ell}\}_{\ell=1}^{N}$ solve the following linear systems over $\GF$ and $\R$:
\begin{align*}
     \sum\nolimits_{\ell=1}^N (\varphi \circ \sgn)( (\mathbf{u}_{\ell})_{i_{\ell}}) &= (\varphi \circ \sgn) \left(\cU_{(i_1, i_2, \dots, i_N)} \right) \mod 2\\
     \sum\nolimits_{\ell=1}^N (\log \circ \, \text{abs}) ((\mathbf{u}_{\ell})_{i_{\ell}} ) &= (\log \circ \, \text{abs}) \left( \cU_{(i_1, i_2, \dots, i_N)} \right)
\end{align*}

We now describe the coefficient matrix of these systems.
Let $\pi: [d]^N \rightarrow [d^N]$ denote a fixed bijection throughout.
Defining $\A_{\pi(i_1, \dots, i_N)} := [ \e_{i_1}, \e_{i_2}, \dots, \e_{i_N}]$ for all $(i_1, \dots, i_N) \in [d]^N$, we have
\begin{align*}
    \A_{\pi(i_1, \dots, i_N)} (\varphi \circ \sgn )(\x) &= (\varphi \circ \sgn) \left(\cU_{(i_1, i_2, \dots, i_N)} \right) \mod 2 \\
    \A_{\pi(i_1, \dots, i_N)} (\log \circ \, \text{abs})(\x) &= (\log \circ \, \text{abs}) \left(\cU_{(i_1, i_2, \dots, i_N)} \right).
\end{align*}
Hence, there is a unique 1-1 tensor-to-vector map $\vect_{\pi}$ such that the above is equivalent to
\begin{align*}
    \A (\varphi \circ \sgn )(\x) &= (\varphi \circ \sgn) \left( \vect_{\pi}\cU \right) \mod 2 \\
    \A (\log \circ \, \text{abs})(\x) &= (\log \circ \, \text{abs}) \left( \vect_{\pi} \cU\right).
\end{align*}
To avoid overloading notation throughout, we let $f := \varphi \circ \sgn \circ \vect_{\pi}$ and $\Tilde{f} := \log \circ 
 \, \text{abs} \circ \vect_{\pi}$, so
\begin{align*}
    \A (\varphi \circ \sgn )(\x) &= f(\cU) \mod 2 \\
    \A (\log \circ \, \text{abs})(\x) &= \Tilde{f} (\cU).
\end{align*}

Notably, $\A$'s rows enumerate all $d^N$ possible vectors of size $dN$ obtained by concatenating $N$ row vectors from $\{\e_k\}_{k \in [d]}$.
This observation enables a simple proof of a result we use extensively---that $\A$ has the same rank considered as a matrix over $\GF$ or $\R$.
\begin{lemma}\label{lemma:rankA}
    The matrix $\A$ satisfies $\rank_{\GF}(\A) = \rank_{\mathbb{R}}(\A) = dN - (N-1)$.
\end{lemma}
\begin{proof}
    See Section~\ref{section:proof-rank-A}.
\end{proof}

 In summary of the previous observations, we have the following.

 \begin{observation}\label{observation:decomposition_lemma}
    For any tensor $\cU \in  \otimes_{i=1}^{N} \mathbb{R}_{\neq 0}^d$,
    \begin{equation}\label{eqn:rank1-iff-consistent}
        \cU \text{ is rank-$1$} \implies \text{linear systems } \system{\A}{f(\cU)}{\GF} \text{ and }\system{\A}{\Tilde{f}(\cU)}{\mathbb{R}} \text{ are both consistent}
    \end{equation}
    Furthermore, if for a rank-$1$ tensor $\mathcal{T} \in \otimes_{i=1}^{N} \mathbb{R}_{\neq 0}^d$ the joint solution sets of
    \begin{equation}\label{eqn:joint-solution-sets}
        \left(\system{\A}{f(\mathcal{T})}{\GF}, \system{\A}{\Tilde{f}(\mathcal{T})}{\mathbb{R}}\right) \text{ and } \left(\system{\A}{f(\cU)}{\GF}, \system{\A}{\Tilde{f}(\cU)}{\mathbb{R}} \right) 
    \end{equation}
    are equivalent, then $\cU = \mathcal{T}$.
 \end{observation}

Indeed,~\eqref{eqn:joint-solution-sets} is the consequence of the following procedure.
Supposing we had access to a solution of $\y_1$ of $\system{\A}{f(\cU)}{\GF}$, and a solution $\y_2$ of $\system{\A}{\Tilde{f}(\cU)}{\mathbb{R}}$, then we could recover the tensor from
\begin{equation*}
    \cU = \vect_{\pi}^{-1}(\varphi^{-1}(\A \y_1) \odot \exp( \A \y_2)),
\end{equation*}
where $\odot$ denotes the Hadamard product.
Therefore, the crux of the issue is accessing these solutions.
Fortunately, as we'll formalize, this is can be achieved by observing sufficiently many random entries of $\cU$.
\\

To preface, for each subset $S \subseteq [d]^N$, we define the \emph{row selection matrix} $\D_S$ as the matrix such that
\begin{equation}\label{eqn:selection-matrix}
    \D_S \A = \begin{pmatrix}
        \vdots \\ \A_{\pi(i_1, \dots, i_N)} \\ \vdots
    \end{pmatrix} \quad \forall (i_1, \dots, i_N) \in S,
\end{equation}
where the rows adhere to the ordering induced by $\pi$.
We similarly define $\D_{\bar{S}}$ for the complement $\bar{S} := [d]^N \setminus S$.
As a result, there is always a row permutation matrix $\P_S \in \{0,1\}^{d^N \times d^N} $ where
    \begin{equation}\label{eqn:permutation}
        \mathbf{P}_{S} \A = \begin{pmatrix}
        \D_{S} \A \\ \D_{\bar{S}}\A 
        \end{pmatrix}.
    \end{equation}
As we've established, since each entry corresponds with a particular row of the augmented systems, we have the next observation.
\begin{observation}\label{observation:isometry}
    Let $\cU \in \NZT$ be rank-1 tensor, and let $S \subseteq [d]^{N}$.
    Then the subset of entries $\{\cU_{(i_1, \dots, i_N)}\}_{(i_1, \dots, i_N) \in S}$ is isomorphic to the pair of linear systems
    \begin{equation}\label{eqn:sub-system}
        \system{\D_{S} \A}{ \D_{S}f(\cU)}{\GF}, \quad \system{\D_{S} \A}{\D_{S}\Tilde{f}(\cU)}{\mathbb{R}}.
    \end{equation}
\end{observation}

Combining Observation~\ref{observation:isometry} and Observation~\ref{observation:decomposition_lemma}, it is immediate that $S$ just needs to satisfy that~\eqref{eqn:sub-system} has the same joint solution set as $ \left(\system{\A}{f(\cU)}{\GF}, \quad \system{\A}{\Tilde{f}(\cU)}{\mathbb{R}} \right)$
to complete the tensor via the aforementioned procedure.

In Section~\ref{section:proof-rank-A}, we show the sufficient conditions are $\Row_{\GF}(\D_S \A) = \Row_{\GF}( \A)$ and $\Row_{\mathbb{R}}(\D_S \A) = \Row_{\mathbb{R}}(\A)$.
In other words, given that these hold, any joint solution to~\eqref{eqn:sub-system} yields a joint solution of the overall systems.
The exact procedure is described in the pseudocode for Algorithm~\ref{pdcode:exactrank1tc}.
\\

\begin{algorithm} 
\caption{Gauss-Jordan-Solver}
\label{pdcode:exactrank1tc}
\begin{algorithmic}[1]
\item[\textbf{Input:}] Tensor entries $\{\cU_{(i_1, \dots, i_N)}\}_{(i_1, \dots, i_N) \in S}$ for $S \subseteq [d]^N$
\vspace{0.5em}
\item[\textbf{Output:}] $\hat{\cU}_{(i_1, \dots, i_N)}$ for each desired $(i_1, \dots, i_N) \in [d]^N$
\vspace{0.5em}
\State $\y_1 \gets $ any solution $\y$ to $\D_{S}\A \y = \D_{S}{f(\cU)}$ over $\mathbb{F}_{2}$
\vspace{0.5em}
\State $\y_2 \gets $ any solution $\y$ to $\D_{S}\A \y = \D_{S}{\Tilde{f}(\cU)}$ over $\mathbb{R}$
\vspace{0.5em}
\State \textbf{return} $\hat{\cU}_{(i_1, \dots, i_N)} = \varphi^{-1}(\inner{\A_{\pi(i_1,\dots,i_N)}}{\y_1}) \exp(\inner{\A_{\pi(i_1,\dots,i_N)}}{\y_2})$ for $(i_1, \dots, i_N) \in [d]^N$
\end{algorithmic}
\end{algorithm}

Hence, we can use Algorithm~\ref{pdcode:exactrank1tc} to solve \textbf{Rank-1 Completion} by simply running it for a randomly drawn $S$.
Clearly, we only need to show these conditions hold with sufficient probability.
We prove that $O ((dN)^2 \log d)$ samples suffice, although we conjecture it can be improved to $ O(dN \log dN)$.
Since the result below is crucial for Theorem~\ref{theorem:main-result}, we overview our proof technique.

\begin{lemma}\label{lemma:rowcollector}
        Let $S$ be the subset of indices induced by $m$ uniformly drawn rows of $\A$, with replacement.
        Then, $m \lesssim (dN)^2 \log d$ samples suffice to ensure both 
        \begin{equation}
            \Row_{\GF}(\D_S\A) = \Row_{\GF}(\A), \quad \Row_{\mathbb{R}}(\D_S \A) = \Row_{\mathbb{R}}(\A)
        \end{equation}
        simultaneously hold with probability $\geq 2/3$.
\end{lemma}
\begin{proof}
    See Section~\ref{sec:rowcollector}.
\end{proof} 

At a high-level, our proof views $\Row_{\GF}(\D_S \A)$ as the result of a sequentially constructed subspace $W$ of $\A$'s rowspace.
We correspond the sample paths of this random process to trajectories on a Markov chain whose states are indexed by $(\dim W, W)$.
Due to the first state coordinate, the chain jumps to a new state no more than $\rank_{\GF}(\A)$ times before hitting the absorbing state $(\rank_{\GF}(\A), \Row_{\GF}(\A))$.
As a result, the measure of the "bad" event $\{\Row_{\GF}(\D_S \A) \neq \Row_{\GF}(\A)\}$ is given by the cumulative measure of the "bad" trajectories, i.e. those that stagnate and never hit the absorbing state.

It turns out to be very easy to prove the chain self-loops with probability $ \leq 1-1/d$.
By pigeonholing, every "bad" trajectory of length $T > \rank_{\GF}(\A)$  has $\Omega(T-\rank_{\GF}(\A))$ self-loops.
Thus, the measure of each "bad" trajectory shrinks exponentially with rate rate $\frac{1}{d}\Omega(T-\rank_{\GF}(\A))$.
A counting argument shows there are no more than $d^{N\rank_{\GF}(\A)}$ "bad" trajectories, from which we show $O(d \log d^{N \rank_{\GF}(\A)}) = O((dN)^2 \log d)$ samples suffice.
To handle the other condition (over $\mathbb{R}$), we show it readily follows from the next well-known fact in tandem with Lemma~\ref{lemma:rankA}.
\begin{fact}\label{fact:r-gf2-rank}
    For any binary matrix $\mathbf{B}$, $\rank_{\GF}(\mathbf{B}) \leq \rank_{\mathbb{R}}(\mathbf{B})$.
\end{fact}

In the remaining sections, we build up towards Theorem~\ref{theorem:main-result}, beginning with establishment of $\A$'s rank in the next section.

\newpage

\section{Proof of Lemma~\ref{lemma:rankA}}\label{section:proof-rank-A}
To give a preview, our strategy is to construct a considerably simpler matrix whose rowspace is identical to $\A$'s.
We then show that this matrix has the claimed rank.
Additionally, the structure of these matrices will be a useful reference to streamline the proofs in the next sections.
We now provide details of the proof.
\\

To start, fix an arbitrary row $ [\e_{i_1}, \dots, \e_{i_N}] \in \{\A_1, \dots, \A_{d^N} \}$, and let $\Phi \in \{0,1\}^{(d-1)N+1 \times dN}$ be the following $\GF$-valued matrix:

\begin{equation}\label{eqn:phi-matrix}
    \Phi := \left( \begin{array}{cccccc}
    \e_{i_1} & \e_{i_2} & \e_{i_3} & \cdots & \e_{i_N} & \\[0.3em]
    \hline \\[-0.8em]
    \e_1 + \e_2 & \textcolor{lightgray}{\o} & \textcolor{lightgray}{\o} & \cdots & \textcolor{lightgray}{\o} & \\
    \e_1 + \e_3 & \textcolor{lightgray}{\o} & \textcolor{lightgray}{\o} & \cdots & \textcolor{lightgray}{\o} & \\
    \vdots & \vdots & \vdots & \ddots & \vdots & \\
    \e_1 + \e_d & \textcolor{lightgray}{\o} & \textcolor{lightgray}{\o} & \cdots & \textcolor{lightgray}{\o} & \\[0.3em]
    \hline \\[-0.8em]
    \textcolor{lightgray}{\o} & \e_1 + \e_2 & \textcolor{lightgray}{\o} & \cdots & \textcolor{lightgray}{\o} & \\
    \textcolor{lightgray}{\o} & \e_1 + \e_3 & \textcolor{lightgray}{\o} & \cdots & \textcolor{lightgray}{\o} & \\
    \vdots & \vdots & \vdots & \ddots & \vdots & \\
    \textcolor{lightgray}{\o} & \e_1 + \e_d & \textcolor{lightgray}{\o} & \cdots & \textcolor{lightgray}{\o} & \\[0.3em]
    \hline \\[-0.8em]
    \textcolor{lightgray}{\o} & \textcolor{lightgray}{\o} & \e_1 + \e_2 & \cdots & \textcolor{lightgray}{\o} & \\
    \textcolor{lightgray}{\o} & \textcolor{lightgray}{\o} & \e_1 + \e_3 & \cdots & \textcolor{lightgray}{\o} & \\
    \vdots & \vdots & \vdots & \ddots & \vdots & \\
    \textcolor{lightgray}{\o} & \textcolor{lightgray}{\o} & \e_1 + \e_d & \cdots & \textcolor{lightgray}{\o} & \\[0.3em]
    \hline \\[-0.8em]
    \vdots & \vdots & \vdots & \ddots & \vdots & \\[0.3em]
    \hline \\[-0.8em]
    \textcolor{lightgray}{\o} & \textcolor{lightgray}{\o} & \textcolor{lightgray}{\o} & \cdots & \e_1 + \e_2 & \\
    \textcolor{lightgray}{\o} & \textcolor{lightgray}{\o} & \textcolor{lightgray}{\o} & \cdots & \e_1 + \e_3 & \\
    \vdots & \vdots & \vdots & \ddots & \vdots & \\
    \textcolor{lightgray}{\o} & \textcolor{lightgray}{\o} & \textcolor{lightgray}{\o} & \cdots & \e_1 + \e_d & \\
    \end{array} \right) := \begin{pmatrix}
        [\e_{i_1}, \dots, \e_{i_N}] \\ \Phi_1 \\ \Phi_2 \\ \vdots \\ \Phi_N
    \end{pmatrix}
\end{equation}

For the next claims, recall that $\A$'s rows consist of all $d^N$ possible vectors of size $dN$ obtained by concatenating $N$ row vectors from $\{\e_k\}_{k \in [d]}$.
\begin{claim}\label{claim:samerowspace}
    $\Row_{\GF}(\Phi) = \Row_{\GF}(\A)$.
\end{claim}
\begin{proof}[Proof of Claim~\ref{claim:samerowspace}]

    By the above assertion, the first row of $\Phi$ is a row of $\A$.
    Consider any other row of $\Phi$.
    This row can be expressed as the sum $[\e_1, \e_1, \dots, \e_1] + [\e_1, \dots, \e_1, \e_{k}, \e_1 \dots, \e_1] \mod 2$ for some $k \in [d]$---both of which are also rows of $\A$.
    Hence, each individual row is contained in $\Row_{\GF}(\A)$, so $\Row_{\GF}(\Phi) \subseteq \Row_{\GF}(\A)$.
    On the other hand, for an arbitrary row $[\e_{j_1}, \dots, \e_{j_N}]$ of $\mathbf{A}$, we can write
    \begin{equation*}
        [\e_{j_1}, \dots, \e_{j_N}] = [\e_{i_1}, \dots, \e_{i_N}] + \sum\nolimits_{k=1}^{N} [\o, \dots, \o, (\e_1 + \e_{i_k}) + (\e_1 + \e_{j_k}), \o, \dots, \o]  \mod 2.
    \end{equation*}
    Each of the summands are evidently in the rowspace of $\Phi$. Hence, $\Row_{\GF}(\A) \subseteq \Row_{\GF}(\Phi)$.
\end{proof}

\begin{claim}\label{claim:phi-rowrank}
    $ \dim \Row_{\GF}(\Phi) = dN-(N-1)$
\end{claim}
\begin{proof}[Proof of Claim~\ref{claim:phi-rowrank}]
    We first establish that each submatrix $\Phi_i$ for $i \in [N]$ has full row rank.
    Without loss of generality, consider the rows of $\Phi_1$.
    Assume by contradiction $\c := (c_1, \dots, c_{d-1}) \neq \mathbf{0}_{d-1} \in \GF^{d-1}$ describes a trivial linear combination of them, i.e.
    \begin{equation*}
        [  \e_1\sum\nolimits_{i=1}^{d-1} c_i + \sum\nolimits_{i=1}^{d-1}c_{i} \e_{i+1}, \o, \dots, \o] = \mathbf{0}_{dN} \mod 2.
    \end{equation*}
    Since no subset of $\e_2, \dots, \e_{d}$ sums to $\o$, $\sum_{i=1}^{d-1}c_{i} \e_{i+1}$ must be identically $\o$, but this necessitates that $\c = \mathbf{0}_{d-1}$, contradiction.
    An analogous argument applies for every submatrix.
    Thus, $\rank_{\GF}(\Phi_1) = \dots = \rank_{\GF}(\Phi_N) = d-1$.
    \\

    It is easy to see that any linear combination of rows in $\Phi_1, \dots, \Phi_N$ result in a nonzero vector.
    Hence, the stacked matrix $(\Phi_1; \Phi_2; \dots; \Phi_N)$ constitutes a linearly independent set of size $(d-1)N$.
    We now elucidate that the first vector is also linearly independent of this stacked matrix.
    To this end, assume by contradiction that there exists some linear combination of rows in $\Phi_1, \dots, \Phi_N$ that sum to $[\e_{i_1}, \dots, \e_{i_N}]$, with coefficients
    \begin{equation*}
        \c := [\c^1, \c^2, \dots, \c^N] = \left((c^1_1, \dots, c^1_{d-1}), \dots, ((c^N_1, \dots, c^N_{d-1}) \right)  \in \GF^{(d-1)N} 
    \end{equation*}
    Picking any $k \in [N]$, the $k^{th}$ \emph{column-block} satisfies
    \begin{equation*}
        \e_1\sum\nolimits_{i=1}^{d-1} c_i^k + \sum\nolimits_{i=1}^{d-1}c_i^k \e_{i+1} = \e_{i_k} \mod 2.
    \end{equation*}
    If $\e_{i_k} = \e_1$, then the first sum dictates $c_1^k, \dots, c_{d-1}^k$ must have odd parity and contributes a single bit overall.
    But then the second sum contributes an odd number of bits.
    This yields a mismatch between the parity on both sides.
    If $\e_{i_k} \neq \e_1$, then second sum must have all $c_i^k$'s as zero except for one of them, but then the first sum contributes a single bit, a contradiction for the same reason as before.
    Thus, the $(d-1)N + 1$ rows of $\Phi$ span a linear space of dimension $dN-(N-1)$.

\end{proof}

Claim~\ref{claim:phi-rowrank} and Claim~\ref{claim:samerowspace} together imply $\Row_{\GF}(\A) = dN-(N-1)$.
For easy reference, we state a trivial corollary following from the fact that $[\e_{i_1}, \dots, \e_{i_N}]$ was initially picked arbitrarily.
\begin{corollary} \label{corollary:helper}
    For any $[\e_{i_1}, \dots, \e_{i_N}] \in \{\A_1, \dots, \A_{d^N}\}$ chosen to construct $\Phi$, the rows of $\Phi$ in~\eqref{eqn:phi-matrix} consist of a basis of $\Row_{\GF}(\A)$.
\end{corollary}

To handle the case of $\mathbb{R}$, one can apply an almost identical proof to establish the analogous claims for a $\R$-valued matrix $\Tilde{\Phi}$, which instead contains rows of the form $[\o, \dots, \o, \e_1 - \e_{i_k}, \o, \dots, \o]$.
We relegate the full matrix description to the Appendix~\ref{appendix:phi-tilde}.
Following this, the proof of Lemma~\ref{lemma:rankA} is complete.

\section{Proof of Lemma~\ref{lemma:rowcollector}} \label{sec:rowcollector}

To prove Lemma~\ref{lemma:rowcollector}, we can view the rowspace of $\D_S \A$ as the the cumulative span of the random variable sequence $Y_1, Y_2, \dots, Y_m$ where $Y_t \stackrel{iid}{\sim} \text{Unif}(\{\A_1, \dots, \A_{d^N}\})$.
\\

Importantly, by Fact~\ref{fact:r-gf2-rank} and Lemma~\ref{lemma:rankA}, whenever we have $\Row_{\GF}(\D_S \A) = \Row_{\GF}(\A)$ we also have $\Row_{\mathbb{R}}(\D_S \A) = \Row_{\mathbb{R}}(\A)$.
Therefore, to prove our choice of $m$ suffices it is enough to show $\dim \Span\{Y_1 \dots Y_m\} = \rank_{\GF}(\A)$ with probability $\geq 2/3$.
Before proceeding to the proof, we assume the next claim holds, which we verify in the sequel.
\begin{claim}\label{claim:maximal-subset}
    Suppose $W$ is a subspace of $\Row_{\GF}(\A)$ and $W$ contains at least one element of $\{\A_1, \dots, \A_{d^N}\}$.
    If $\dim W < \rank_{\GF}(\A)$, then there are at least $d^{N-1}$ rows of $\A$ which are each linearly independent of $W$.
\end{claim}
The main message of the above is that as long as $W$ is "missing a direction" in $\A$'s rowspace, there are at least $1/d$ fraction of rows that would increase its dimensionality.

\begin{claim}\label{claim:markovchain}
    Let $Y_1, Y_2, \dots$ where $Y_t \stackrel{iid}{\sim} \text{Unif}(\{\A_1, \dots, \A_{d^N}\})$.
    We have that $m \lesssim (dN)^2 \log d$ samples suffice to ensure $\dim \Span\{Y_1, \dots, Y_m\} = \rank_{\GF}(\A)$ with probability $\geq 2/3$.
\end{claim}

\begin{proof}[Proof of Claim~\ref{claim:markovchain}]
Let $m$ be a positive integer.
To each sequence $y = (\y_1, \dots, \y_m) \in \{\A_1, \dots, \A_{d^N}\}^{m}$ we associate another sequence $h^{y} = (h^y_1, \dots, h_m^y)$ where $h^y_t := \Span\{\y_1, \dots, \y_t\}$ is the cumulative span of the first $t$ vectors of the sequence $y$. 
Consider the directed graph $G = (V,E)$, in which
\begin{equation*}
    V := \left\{(\alpha, W) \mid  \alpha \in [\rank_{\GF}(\A)], \; W \text{ is a subspace of } \Row_{\GF}(\A) \right\}.
\end{equation*}
The edgeset $E$ is defined as follows.
For each $m \in \mathbb{N}_{>0}$ and each $\y \in \{\A_1, \dots, \A_{d^N}\}^{m}$, we include in $E$ the directed edge
\begin{equation*}
    \left((\dim h^{y}_t, h^{y}_t),  (\dim h^{y}_{t+1}, h^{y}_{t+1}) \right) \in V \times V
\end{equation*}
for all $t \in [m-1]$, emphasizing that self-loops are allowed.\footnote{Plainly stated, the construction places paths on the graph tracking the cumulative span and its dimension for every possible sequence.}
Notably, $V$ includes the vector $v^{*} := (\rank_{\GF}(\A), \Row_{\GF}(\A))$ because $\A$ trivially has a sequence of $\rank_{\GF}(\A)$ vectors whose span is $\Row_{\GF}(\A)$.
In particular, $v^{*}$ has only a single outgoing edge, which is also a self-loop.
\\

It follows that the joint distribution of $Y_1, Y_2, \dots$ defines a time-homogeneous Markov chain $X_1, X_2, \dots $ on the state space $V$.
For each $(v_1, \dots, v_m) \in \text{Paths}(G)$, denote $(\Tilde{v}_1, \dots, \Tilde{v}_{\Tilde{m}})$ as the truncated path obtained by removing self-loops, i.e. any vertex whose previous vertex is identical to itself. 
It follows that the probability $\dim \Span \{Y_1, \dots, Y_m \} \neq \rank_{\GF}(\A)$ is at most

\begin{align*}
    \mathbb{P}(X_m \neq v^{*}) &= \sum\nolimits_{(v_1, \dots, v_m) \in \text{Paths}(G) : v_m \neq v^{*}} \mathbb{P}(X_1 = v_1, \dots X_m = v_m) \\
    &\stackrel{(a)}{=} \sum\nolimits_{(v_1, \dots, v_m) \in \text{Paths}(G) : v_k \neq v^{*}, \forall k\in [m]} \mathbb{P}(X_1 = v_1, \dots X_m = v_m) \\
    &\stackrel{(b)}{=} \sum\nolimits_{(v_1, \dots, v_m) \in \text{Paths}(G): v_k \neq v^{*}, \forall k\in [m]} \mathbb{P}(X_1 = \Tilde{v}_1, \dots X_{\Tilde{m}} = \Tilde{v}_{\Tilde{m}}) \times (\text{self-loop probabilities}) \\
    &\stackrel{(c)}{\leq} \sum\nolimits_{(v_1, \dots, v_m) \in \text{Paths}(G): v_k \neq v^{*}, \forall k\in [m]}  \left(1-\frac{1}{d} \right)^{(\text{\# self-loops in $(v_1, \dots, v_m)$})} \\
    % &\stackrel{(e)}{\leq} \sum_{(v_1, \dots, v_m): v_k \neq v^{*}}  \left(1-\frac{1}{d} \right)^{(\text{\# self-loops in $(v_1, \dots, v_m)$})} \times \mathbb{I}[\text{(4.a) and (4.b) hold}] \\
    &\stackrel{(d)}{\leq}  \sum\nolimits_{(v_1, \dots, v_m) \in \text{Paths}(G): v_k \neq v^{*}, \forall k\in [m]} e^{-\frac{m-\rank_{\GF}(\A)+1}{d}}\\
    &\stackrel{(e)}{\leq}   e^{-\frac{m-\rank_{\GF}(\A)+1}{d}} d^{N\rank_{\GF}(\A)}\\
    &\stackrel{(f)}{\leq} 1/3
\end{align*}
where $(a)$ follows since $v^{*}$ is an absorbing state; $(b)$ follows by the Markov property and time-homogeneity; $(c)$ uses the observation that $\mathbb{P}(X_{t+1} = v \mid X_t = v) \leq 1-d^{N-1}/d^N$ for $v \in V \setminus \{v^{*}\}$ by Claim~\ref{claim:maximal-subset}; $(d)$ uses the observation that there are $\geq m - \rank_{\GF}(\A) + 1$ self-loops 
% $\{(v_t, v_{t+1})\}_{t=1}^{m-1}$ 
in paths never reaching $v^{*}$, otherwise $> (m-1)-(m-\rank_{\GF}(\A) + 1) = \rank_{\GF}(\A) - 2$ of the edges are associated with an increase of the cumulative span's dimension---implying the chain has to reach $v^{*}$; $(e)$ uses the observation that there are at most $ \leq d^{N \Tilde{m}}$ sequences in $\{\A_1, \dots, \A_{d^N}\}^{\Tilde{m}}$, each of which contributes to at most one length $\Tilde{m}$ loop-less walk in the graph.
Since the path mustn't terminate at $v^{*}$, $\Tilde{m} \leq \rank_{\GF}(\A)$, from which the bound follows.
Finally, $(f)$ follows from choosing $m = \rank_{\GF}(\A) - 1 + \lceil d\log (3d^{N\rank_{\GF}(\A)}) \rceil$, which is $O((dN^2)\log(d))$ since $\rank_{\GF}(\A) = O(dN)$ (c.f. Lemma~\ref{lemma:rankA}).

\end{proof}

Now, it just remains to verify Claim~\ref{claim:maximal-subset}.

\begin{proof}[Proof of Claim~\ref{claim:maximal-subset}]
    Consider the set of vectors
    \begin{equation*}
        \phi_{i}^{n} := [\o, \dots ,\o, \underbrace{\e_1 + \e_i}_{n^{th} \text{ position}}, \o, \dots \o] \in \{0,1\}^{dN}
    \end{equation*}
    which are defined for all $i \in [d]\setminus \{1\}$ and $n \in [N]$.
    There must exist a $\phi_{i}^n$ which is linearly independent of $W$.
    Otherwise $W$ contains the subspace $\Row_{\GF}(\Phi)$ ($\Phi$ is defined in~\eqref{eqn:phi-matrix})---but $\Row_{\GF}(\Phi) = \Row_{\GF}(\A)$ by Corollary~\ref{corollary:helper}, which contradicts the dimension of $W$.
    For this $\phi_{i}^n$, consider the set of unordered vector pairs $\left\{\{\mathbf{a}, \mathbf{b}\} \in \{\A_1, \dots, \A_{d^N}\}^2 \mid \mathbf{a}+\mathbf{b} = \phi_{i}^{n} \mod 2 \right\}$.
    
    Notably, this set contains exactly $d^{N-1}$  pairs, as one must fix the column-block in the $n^{th}$ position to be $\e_{1}$ for one, which fixes the other to be $\e_{i}$---varying over the last $N-1$ column-blocks with $d$ choices for each.
    For each pair, at least one of $\mathbf{a}$ or $\mathbf{b}$ must be linearly independent of $W$, for otherwise it contradicts that $\phi_{i}^{n}$ is not in the subspace $W$.
    Hence, one can find $\geq d^{N-1}$ rows of $\A$, each individually linearly independent of $W$.
\end{proof}

In the next section, we prove the correctness of Algorithm~\ref{pdcode:exactrank1tc}.

\section{Proof of Theorem~\ref{theorem:main-result}, Upper Bound}\label{section:main-result-ub}
Let $m$ denote the quantity in Claim~\ref{claim:markovchain} of the previous section.
The algorithm asserted in Theorem~\ref{section:main-result-ub} simply draws $m' := \max\{m, dN\} = O((dN)^2\log d)$ samples, lets $S$ be the set of associated indices, and runs Algorithm~\ref{pdcode:exactrank1tc} for this input $S$.

The runtime is immediate.
Indeed, the systems in steps (1) and (2) are size $m' \times dN$ and consistent, so Gauss-Jordan on each terminates in time $O(m'(dN)^2)$.
The algorithm takes an additional $O(N)$ time per queried entry in step (3), thus $O(qN + m'(dN)^2)$ overall.

We now turn towards establishing the correctness of Algorithm~\ref{pdcode:exactrank1tc}.
As our choice of $m'$ satisfies Lemma~\eqref{lemma:rowcollector}, we have that $\Row_{\GF}(\D_S\A) = \Row_{\GF}(\A)$ and $ \Row_{\mathbb{R}}(\D_S \A) = \Row_{\mathbb{R}}(\A)$
holds with probability $\geq 2/3$.
Thus, it suffices to prove the following.

\begin{claim}\label{claim:correctness}
    Let $\cU \in \NZT$ be a rank-1 tensor.
    Assume the input $S$ satisfies both
    \begin{equation*}
        \Row_{\GF}(\D_S\A) = \Row_{\GF}(\A), \quad \Row_{\mathbb{R}}(\D_S \A) = \Row_{\mathbb{R}}(\A).
    \end{equation*}
    Then, the output of Algorithm~\ref{pdcode:exactrank1tc} satisfies $\hat{\cU} = \cU$.
\end{claim}

To start, we establish a useful helper lemma.
\begin{lemma}\label{lemma:partial_sys_consistency}
    Suppose $\A \x = \b$ is a consistent linear system over a field $\mathbb{K}$, for which
    \[
    \A = \begin{pmatrix}
    \A_1 \\ \A_2
    \end{pmatrix}, \quad \mathbf{b} = \begin{pmatrix}
    \mathbf{b}_1 \\ \mathbf{b}_2
    \end{pmatrix}.
    \]
    
    Suppose $\Row_{\mathbb{K}}(\A_1) = \Row_{\mathbb{K}} (\A)$.
    If $\x^{*}$ satisfies $\A_1 \x^{*} = \mathbf{b}_1$, then $\A_2 \x^{*} = \mathbf{b}_2$ holds (and evidently $\A \x^{*} = \b)$.
\end{lemma}
\begin{proof}[Proof of Lemma~\ref{lemma:partial_sys_consistency}]
    By consistency there exists a $\mathbf{z}$ such that $\mathbf{A}_1 \mathbf{z} = \mathbf{b}_1$ and $\mathbf{A}_2 \mathbf{z} = \mathbf{b}_2$.
    Since $\x^{*}$ satisfies $\mathbf{A}_1 \x^{*} = \mathbf{b}_1$, we have $\x^{*}-\mathbf{z} \in \ker_{\mathbb{K}} (\mathbf{A}_1) = \Row_{\mathbb{K}}(\mathbf{A}_1)^{\perp} = \Row_{\mathbb{K}}(\mathbf{A})^{\perp} = \ker_{\mathbb{K}}(\mathbf{A})$.
    Hence, $\mathbf{A}(\x^{*} -\mathbf{z}) = \mathbf{0}$, i.e. $\A \x^{*} = \A \mathbf{z}$
    implying $\mathbf{A}_2 \x^{*} = \mathbf{A}_2 \mathbf{z} = \mathbf{b}_2$.
\end{proof}

Now, we have all the tools to prove Claim~\ref{claim:correctness}.

\begin{proof}[Proof of Claim~\ref{claim:correctness}]
    From~\eqref{eqn:permutation}, let $\mathbf{P}_{S} \in \{0,1\}^{d^N \times d^N}$ be a permutation matrix such that

    \begin{equation} \label{eqn:permuted-systems}
        \mathbf{P}_{S} \A = \begin{pmatrix}
        \D_{S} \A \\ \D_{\bar{S}}\A 
    \end{pmatrix}, \quad
    \mathbf{P}_{S} f(\cU) = \begin{pmatrix}
        \D_{S} f(\cU) \\ \D_{\bar{S}} f(\cU)
    \end{pmatrix}, \quad
    \mathbf{P}_{S} \Tilde{f}(\cU) = \begin{pmatrix}
        \D_{S} \Tilde{f}(\cU) \\ \D_{\bar{S}} \Tilde{f}(\cU)
    \end{pmatrix}.
    \end{equation}

    By Observation~\ref{observation:decomposition_lemma} both systems $\system{\D_S \A}{\D_S f(\cU)}{\GF}$ and $\system{\D_S \A}{\D_S \Tilde{f}(\cU)}{\mathbb{R}}$ are consistent---hence steps (1) and (2) return a $\y_1$ and $\y_2$ for which $\D_{S}\A \y_1 = \D_{S} f(\cU)$ and $\D_{S}\A \y_2 = \D_{S} \Tilde{f}(\cU)$.
    
    We invoke Lemma~\eqref{lemma:partial_sys_consistency} to the two systems induced by~\eqref{eqn:permuted-systems}.
    Specifically, we assign $\A \leftarrow \P_S \A$, $\b \leftarrow \P_S f(\cU)$ and use that $\D_{S}\A \y_1 = \D_{S} f(\cU)$, upon which the lemma lets us conclude $\P_S \A \y_1 = \P_S  f(\cU)$.
    Had we instead taken $\b \leftarrow \P_S \Tilde{f}(\cU)$ and used $\D_{S}\A \y_2 = \D_{S} \Tilde{f}(\cU)$, we would conclude $\P_S \A \y_2 = \P_S  \Tilde{f}(\cU)$.
    Since $\P_S$ is a row permutation, this implies
    \begin{equation*}
        \A \y_1 = f(\cU), \quad \A \y_2 =  \Tilde{f}(\cU).
    \end{equation*}
    Hence, recalling the map $\varphi$ and the definition of $f$ and $\Tilde{f}$,
    \begin{align*}
        \varphi^{-1}(\A \y_1) &= (\sgn \circ \vect_{\pi})(\cU) \\
        \exp(\A \y_2) &= (\text{abs} \circ \vect_{\pi})(\cU).
    \end{align*}
    Evidently,
    \begin{equation*}
        \cU_{(i_1, \dots, i_N)} = \varphi^{-1}(\inner{\A_{\pi(i_1,\dots,i_N)}}{\y_1}) \exp\left( \inner{\A_{\pi(i_1,\dots,i_N)}}{\y_2}\right) = \hat{\cU}_{(i_1, \dots, i_N)},
    \end{equation*}
    which is what we wanted to show.
\end{proof}

\section{Proof of Theorem~\ref{theorem:main-result}, Lower Bound}

As we'll describe, the lower bound in Theorem~\ref{theorem:main-result} is a simple consequence of the following.

\begin{lemma}\label{lemma:parallel_ccp}
    Consider a variant of the coupon collector problem in which there are $N \in \mathbb{N}_{>0}$ urns, each containing $d \in \mathbb{N}_{>0}$ unique balls.
    Suppose each draw is given by a uniformly random choice of one ball from each and every urn, for a total of $N$ balls per draw.
    
    \smallskip
    There exists absolute constants $n_0 \in \mathbb{N}$ and $C>0$ such that $dN \geq n_0$ implies that if less than $C d \log dN$ draws are taken, then at least one ball is missed with probability $\geq 2/3$.
    
\end{lemma}

\begin{remark}
    By the classic variant, it is easy to see $\Omega(d\log d)$ is necessary.
    Since our main result strives for optimal $d$ dependence, this would be enough.
    For this reason we leave the proof to Appendix~\ref{appendix:parallel_ccp}, which involves recursively applying Hoeffding's lemma to a particular martingale sequence.
    However, we feel this result clarifies the "coupon collector effect" frequently referred to in the tensor completion literature---often as a remark used to justify the presence of logarithmic factors in the upper bounds.
    In contrast, our lower bound explicitly uses such an argument.
    
\end{remark}

We now provide details of the proof.
\\
\begin{proof}[Proof of Lower Bound, Theorem~\ref{theorem:main-result}]
Let $\mathbf{u}_1, \dots, \mathbf{u}_N \stackrel{iid}{\sim} \text{Unif}(\{+1,-1\}^{d})$ and consider the tensor $\cU := \rho \, (\mathbf{u}_1 \otimes \dots 
\otimes \mathbf{u}_N)$ which serves as input to the algorithm.
 Recalling each entry of $\cU$ is dependent on $N$ out of $dN$ variables (c.f.\eqref{eqn:polynomial}), we say a sampled entry $\cU_{(i_1, \dots , i_N)}$ \emph{collects} the variable $(\mathbf{u}_{k})_{\ell}$ if the former is dependent on the latter.
 We can correspond the samples with the coupon collecting procedure in Lemma~\ref{lemma:parallel_ccp}.
 Concretely, we assign each $\mathbf{u}_i$ to the urn $i$, and each coordinate variable $(\mathbf{u}_i)_j$ with the $j^{th}$ ball in the $i^{th}$ urn.

 Let $m$ denote the quantity indicated by (the proof of) Lemma~\ref{lemma:parallel_ccp}.
 Suppose $\hat{\cU}$ is the output of some algorithm drawing less than $m$ samples.
 By Lemma~\ref{lemma:parallel_ccp} above, the algorithm will not \emph{collect} some variable $(\mathbf{u}_{k'})_{\ell'}$ with probability $\geq 2/3$.
 Let $\mathrm{i}'$ denote any index such that $\cU_{\mathrm{i}'}$ depends on the variable $(\mathbf{u}_{k'})_{\ell'}$.
 Even in the case the estimator knows the value of all the variables except $(\mathbf{u}_{k'})_{\ell'}$, the unobserved entry $\cU_{\mathrm{i}}$ is $\pm \rho$ with probability $1/2$ each.
 Thus, by pigeonholing 
 we have that $|\hat{\cU}_{\mathrm{i}'}-\cU_{\mathrm{i}'}| \geq \frac{1}{2} \, 2\rho$ on at least one of the events $(\mathbf{u}_{k'})_{\ell'}$ is $+1$ or $-1$.
 Since there are $d^{N-1}$ such indices, $\| \hat{\cU}-\cU\|_{F}^2 \geq \rho d^{N-1}$ with probability $1/2$, i.e. $\|\hat{\cU}-\cU\|_F \geq \rho \sqrt{d^{N-1}}$ with probability $ \geq \frac{2}{3} \cdot \frac{1}{2} = \frac{1}{3}$, as was to be shown.
\end{proof}

\section{Open Questions}

As asserted, we believe the upper bound can be improved to $O(d \log d)$ to match the lower bound in the $N = \Theta(1)$ case, either by a refinement of our "matrix sketch over $\GF$" viewpoint, or a different technique altogether.
A more difficult question is whether $\Theta(dN \log dN)$ is the 
correct sampling threshold for arbitrary $N$.
Indeed, given that $\A$ is nonnegative, then one only needs to sketch the rows of $\A$ as a real matrix.
Hence, one can apply  well-known results from leverage score sampling (e.g.~\cite{cohen2015uniform}) to yield a $O(dN \log dN)$ upper bound.

\bibliographystyle{alpha}
\bibliography{references}

\newpage
\appendix

\section{Appendix A : Proof of Lemma~\ref{lemma:parallel_ccp}} \label{appendix:parallel_ccp}
In this section, we prove Lemma~\ref{lemma:parallel_ccp}.
     \begin{itemize}
    \item Let $Z_t^{i}$ for $i \in [N]$ denote the random variable counting the remaining balls in urn $i$ yet to be seen by and including the $t^{th}$ draw.
    \item Let $I_t^{i}$ for $i \in [N]$ denote the indicator random variable which is '1' if the $t^{th}$ draw sees a previously unseen ball in urn $i$, and is '0' otherwise.
    \item Let $\mathcal{F}_{t}$ denote the natural filtration generated by the random variables $\{Z_{1}^{i}, Z_{2}^{i}, \dots, Z_{t}^{i} \}_{i=1}^{N}$.
    % \textcolor{brown}{maybe explain what you mean by filtration here.}
\end{itemize}
Finally, denote $Z_t := \sum_{i=1}^{N} Z_t^{i}$ and $I_t := \sum_{i=1}I_t^{i}$.
In what follows, it is helpful to note that $Z_{t+1} = Z_{t} - I_{t+1}$.
\\
\begin{claim}\label{claim:martinagle}
    For any $s > 0$, $d \geq 2$, and $t \in \mathbb{N}_{> 0}$, we have
    \begin{equation}
     \mathbb{E}(e^{-sZ_{t+1}}) \leq e^{\frac{s^2 dN}{8} -s(1-\frac{1}{d})^{t}\frac{dN}{2}}
\end{equation}
\end{claim}
\begin{proof}
    Denoting $\alpha := 1- \frac{1}{d}$,
    \begin{align*}
    \mathbb{E}(e^{-sZ_{t+1}}) &= \mathbb{E}(\mathbb{E}(e^{-s(Z_t - I_{t+1})} \mid \mathcal{F}_t))
    \\
    &= \mathbb{E}(e^{-sZ_t} \mathbb{E}(e^{sI_{t+1}} \mid \mathcal{F}_t)) \\
    &\stackrel{(a)}{=} \mathbb{E}(e^{-sZ_t} \prod_{i=1}^{N}\mathbb{E}(e^{sI_{t+1}^{i}} \mid \mathcal{F}_t))   \\
    &\stackrel{(b)}{\leq} \mathbb{E} \left(e^{-sZ_t} \prod_{i=1}^{N} e^{s\mathbb{E}(I_{t+1}^{i}\mid \mathcal{F}_t)) + \frac{s^2}{8}}   \right) \\
    &\stackrel{(c)}{=} e^{\frac{s^2N}{8}} \mathbb{E} \left(e^{-sZ_t} \prod_{i=1}^{N} e^{s \frac{Z_{t}^{i}}{d}}   \right) \\
    &= e^{\frac{s^2N}{8}} \mathbb{E} \left(e^{-sZ_t} e^{s \frac{Z_{t}}{d}}   \right) \\
    &= e^{\frac{s^2N}{8}} \mathbb{E} \left(e^{-s\alpha Z_t}  \right) \\
    &\stackrel{(d)}{\leq} e^{\frac{s^2N}{8}} \left(e^{\frac{\alpha^2 s^2N}{8}} \mathbb{E} \left(e^{-s\alpha^2 Z_{t-1}}  \right) \right)  \\
    &\leq \dots \\
    &\leq e^{\frac{s^2N}{8}(1+\alpha^2 + \alpha^4 +  \dots + \alpha^{2(t-1)})} \mathbb{E} \left(e^{-s\alpha^{t} Z_1}  \right)
    \\
    &\stackrel{(e)}{=} e^{\frac{s^2N}{8}(1+\alpha^2 + \alpha^4+ \dots + \alpha^{2(t-1)})} e^{-s\alpha^{t}(d-1)N}
    \\
    &\leq e^{\frac{s^2N}{8}(1+\alpha + \alpha^2+ \dots + \alpha^{(t-1)})} e^{-s\alpha^{t}(d-1)N}
    \\
    &\leq e^{\frac{s^2 dN}{8} -s\alpha^{t}\frac{dN}{2}}
\end{align*}
where $(a)$ follows since each urn is sampled from independently; $(b)$ applies Hoeffding's lemma; $(c)$ uses the observation that, given the filtration up to time $t$, the $i^{th}$ urn at time $t+1$ "sees" a new ball if sampling one of $Z_{t}^{i}$ uncollected balls out of $d$; $(d)$ is the first recursive application of the bound; and $(e)$ uses the simple observation that the first draw always "sees" $N$ balls, so $Z_1 = (d-1)N$ almost surely.

\end{proof}
We are now in a position to prove Lemma~\ref{lemma:parallel_ccp}.

\begin{proof}[Proof of Lemma~\ref{lemma:parallel_ccp}]
    Fix $\beta \in (0,1)$ to be decided later.
    Assume that the following holds
    \begin{equation} \label{ineq:deficite-samples}
        t \leq \beta d\log d N.
    \end{equation}
    We have $(1-\frac{1}{d})^{t} \gtrsim e^{-\frac{t}{d}} \geq (dN)^{-\beta}$, so that for any $s > 0$ and $\epsilon > 0$, by Claim~\ref{claim:martinagle},
    \begin{equation*}
        \mathbb{P}(Z_{t+1} \leq \epsilon) \leq e^{s\epsilon}\mathbb{E}(e^{-sZ_{t+1}}) \leq e^{s \epsilon + \frac{s^2 dN}{8} -s(1-\frac{1}{d})^{t}\frac{dN}{2}} \lesssim e^{s\epsilon + \frac{s^2 dN}{8} - \frac{s}{2} (dN)^{1-\beta}}.
    \end{equation*}

    Let us constrain $\epsilon \in (0, \frac{1}{2}(dN)^{1-\beta})$ and take $s = \frac{4}{dN} ( \frac{1}{2}(dN)^{1-\beta} - \epsilon)$ to give
    \begin{equation*}
        \mathbb{P}(Z_{t+1} \leq \epsilon) \lesssim \exp\left(- \frac{2}{dN} \left( \frac{1}{2}(dN)^{1-\beta}-\epsilon \right)^2 \right).
    \end{equation*}
    Suppose $\epsilon= \frac{1}{4}(dN)^{1-\beta}$ and $dN \geq 7$ so that 
    \begin{equation*}
        \mathbb{P}(Z_{t+1} \leq \frac{1}{4} (dN)^{1-\beta}) \leq \exp\left(-\frac{(dN)^{1-2\beta}}{8} \right) .
    \end{equation*}
    In particular, for $\beta = 1/4$ and $dN \geq 78$, this implies
    \begin{equation*}
        \mathbb{P}(Z_{t+1} \leq 1) \leq \mathbb{P}(Z_{t+1} \leq \frac{1}{4} (dN)^{3/4}) \leq \exp\left(-\frac{\sqrt{dN}}{8} \right) \leq \frac{1}{3}.
    \end{equation*}
    Thus, $\mathbb{P}(Z_{t+1} > 1 ) \geq \frac{2}{3}$.
    In other words, with probability $\geq 2/3$ there remains an unseen ball if~\eqref{ineq:deficite-samples} holds for $\beta = 1/4$ and $dN \geq 78$.
\end{proof}

\section{Appendix B: $\Tilde{\Phi}$ for Proof of Lemma~\ref{lemma:rankA}}\label{appendix:phi-tilde}

\bigskip

\begin{equation*}\label{eqn:phi-matrix-real}
    \Tilde{\Phi} := \left( \begin{array}{cccccc}
    \e_{i_1} & \e_{i_2} & \e_{i_3} & \cdots & \e_{i_N} & \\[0.3em]
    \hline \\[-0.8em]
    \e_1 - \e_2 & \textcolor{lightgray}{\o} & \textcolor{lightgray}{\o} & \cdots & \textcolor{lightgray}{\o} & \\
    \e_1 - \e_3 & \textcolor{lightgray}{\o} & \textcolor{lightgray}{\o} & \cdots & \textcolor{lightgray}{\o} & \\
    \vdots & \vdots & \vdots & \ddots & \vdots & \\
    \e_1 - \e_d & \textcolor{lightgray}{\o} & \textcolor{lightgray}{\o} & \cdots & \textcolor{lightgray}{\o} & \\[0.3em]
    \hline \\[-0.8em]
    \textcolor{lightgray}{\o} & \e_1 - \e_2 & \textcolor{lightgray}{\o} & \cdots & \textcolor{lightgray}{\o} & \\
    \textcolor{lightgray}{\o} & \e_1 - \e_3 & \textcolor{lightgray}{\o} & \cdots & \textcolor{lightgray}{\o} & \\
    \vdots & \vdots & \vdots & \ddots & \vdots & \\
    \textcolor{lightgray}{\o} & \e_1 - \e_d & \textcolor{lightgray}{\o} & \cdots & \textcolor{lightgray}{\o} & \\[0.3em]
    \hline \\[-0.8em]
    \textcolor{lightgray}{\o} & \textcolor{lightgray}{\o} & \e_1 - \e_2 & \cdots & \textcolor{lightgray}{\o} & \\
    \textcolor{lightgray}{\o} & \textcolor{lightgray}{\o} & \e_1 - \e_3 & \cdots & \textcolor{lightgray}{\o} & \\
    \vdots & \vdots & \vdots & \ddots & \vdots & \\
    \textcolor{lightgray}{\o} & \textcolor{lightgray}{\o} & \e_1 - \e_d & \cdots & \textcolor{lightgray}{\o} & \\[0.3em]
    \hline \\[-0.8em]
    \vdots & \vdots & \vdots & \ddots & \vdots & \\[0.3em]
    \hline \\[-0.8em]
    \textcolor{lightgray}{\o} & \textcolor{lightgray}{\o} & \textcolor{lightgray}{\o} & \cdots & \e_1 - \e_2 & \\
    \textcolor{lightgray}{\o} & \textcolor{lightgray}{\o} & \textcolor{lightgray}{\o} & \cdots & \e_1 - \e_3 & \\
    \vdots & \vdots & \vdots & \ddots & \vdots & \\
    \textcolor{lightgray}{\o} & \textcolor{lightgray}{\o} & \textcolor{lightgray}{\o} & \cdots & \e_1 - \e_d & \\
    \end{array} \right) := \begin{pmatrix}
        [\e_{i_1}, \dots, \e_{i_N}] \\ \Phi_1 \\ \Phi_2 \\ \vdots \\ \Phi_N
    \end{pmatrix}
\end{equation*}

\end{document}